\documentclass{vgtc}                          

\ifpdf
  \pdfoutput=1\relax                   
  \usepackage{graphicx}                
  \DeclareGraphicsExtensions{.pdf,.png,.jpg,.jpeg} 
\else
  \usepackage{graphicx}                
  \DeclareGraphicsExtensions{.eps}     
\fi%

\graphicspath{{figures/}{pictures/}{images/}{./}} 

\usepackage{microtype}                 
\PassOptionsToPackage{warn}{textcomp}  
\usepackage{textcomp}                  
\usepackage{mathptmx}                  
\usepackage{times}                     
\usepackage{cite}                      
\usepackage{tabu}                      
\usepackage{booktabs}                  

\usepackage{amsmath,hyperref,lineno,amssymb,amsthm,mathtools,url,verbatim}
\usepackage{nccmath}
\usepackage{pgfplots}
\pgfplotsset{grid style={red}}
\usepackage{filecontents}
\usepackage{microtype}
\usepackage{float}
\usepackage{tikz}
\usetikzlibrary{matrix}
\usepackage{tikz-cd}
\usepackage{algorithmic, algorithm2e}
\usepackage[inline]{enumitem}
\usepackage{tabularx}
\usetikzlibrary{decorations.markings}

\nolinenumbers 

\usepackage{verbatim,amsfonts,amsmath,amssymb}
\usepackage{hyperref}
\usepackage{algorithmic}

\usepackage{tikz}
\usetikzlibrary{matrix}
\usetikzlibrary{arrows,positioning,shapes.geometric}
\usetikzlibrary{calc,trees,positioning,arrows,chains,shapes.geometric,%
    decorations.pathreplacing,decorations.pathmorphing,shapes,%
   shapes.symbols, trees,backgrounds}
      \tikzstyle{blockzm1} = [rectangle, draw, fill=white, 
    text width=1.5em, text centered, rounded corners, minimum height=1.5em, minimum width = 1.5em]
    
      \tikzstyle{blockzm2} = [rectangle, draw, fill=white, 
    text width=1.0em, text centered, rounded corners, minimum height=1.0em, minimum width = 1.0em]
    
      \tikzstyle{blockzm1r} = [rectangle, draw, fill=red!70, 
    text width=1.5em, text centered, rounded corners, minimum height=1.5em, minimum width = 1.5em]
      \tikzstyle{blockzm1o} = [rectangle, draw, fill=orange!80, 
    text width=1.5em, text centered, rounded corners, minimum height=1.5em, minimum width = 1.5em]
      \tikzstyle{blockzm1y} = [rectangle, draw, fill=yellow!50, 
    text width=1.5em, text centered, rounded corners, minimum height=1.5em, minimum width = 1.5em]
      \tikzstyle{blockzm1g} = [rectangle, draw, fill=green!60, 
    text width=1.5em, text centered, rounded corners, minimum height=1.5em, minimum width = 1.5em]
        \tikzstyle{blockzm1p} = [rectangle, draw, fill=blue!50, 
    text width=1.5em, text centered, rounded corners, minimum height=1.5em, minimum width = 1.5em]
    
   \tikzstyle{blockz} = [rectangle, draw, fill=white, 
    text width=1.75em, text centered, rounded corners, minimum height=1.75em, minimum width = 1.75em]
     \tikzstyle{blockzflexi} = [rectangle, draw, fill=white, 
    text centered, rounded corners]
    \tikzstyle{blockz2} = [rectangle, draw, fill=white, 
    text width=2em, text centered, rounded corners, minimum height=2em, minimum width = 2em]
        \tikzstyle{blockzL} = [rectangle, draw, fill=white, 
    text width=4em, text centered, rounded corners, minimum height=3em, minimum width = 4em]
    
       \tikzstyle{vertex}=[circle,fill=black!25,minimum size=20pt,inner sep=0pt]
\tikzstyle{selected vertex} = [vertex, fill=red!24]
\tikzstyle{edge} = [draw,thick,-]
\tikzstyle{weight} = [font=\small]
\tikzstyle{selected edge} = [draw,line width=5pt,-,red!50]
\tikzstyle{ignored edge} = [draw,line width=5pt,-,black!20]
    
\tikzset{
  treenode/.style = {align=center, inner sep=0pt, text centered,
    font=\sffamily},
  arn_n/.style = {treenode, circle, white, font=\sffamily\bfseries, draw=black,
    fill=black, text width=1.5em},
  arn_r/.style = {treenode, circle, red, draw=red, 
    text width=1.5em, very thick},
  arn_x/.style = {treenode, rectangle, draw=black,
    minimum width=0.5em, minimum height=0.5em}
}
\usepackage{tikz-cd}
\usepackage{verbatim}
\usepackage{subcaption}

\usepackage{tabularx}
\usepackage{array} 
\usepackage{varwidth} 
\newcolumntype{M}{>{\begin{varwidth}{3cm}}l<{\end{varwidth}}}

\numberwithin{algocf}{section}

\theoremstyle{definition}
\newtheorem{dfn}[algocf]{Definition}
\newtheorem{prob}[algocf]{Problem}

\theoremstyle{definition}
\newtheorem{exa}[algocf]{Example}
\theoremstyle{definition}
\newtheorem{rem}[algocf]{Remark}
\newtheorem{cexa}[algocf]{Counterexample}
\theoremstyle{plain}

\newtheorem{lem}[algocf]{Lemma}

\usepackage{mathtools}
\DeclarePairedDelimiter\ceil{\lceil}{\rceil}

\newcommand{\C}{\mathcal{C}}

\newcommand{\Child}{\mathrm{Children}}
\newcommand{\Desc}{\mathrm{Descendants}}

\newcommand{\R}{\mathbb{R}}

\newcommand{\Z}{\mathbb{Z}}

\newcommand{\I}{\mathcal{I}}
\newcommand{\T}{\mathcal{E}}

\newcommand{\bs}{\hfill $\blacksquare$}

\usepackage{array} 
\usepackage{varwidth} 
\newcolumntype{M}{>{\begin{varwidth}{3cm}}l<{\end{varwidth}}}

\onlineid{0}

\vgtccategory{Research}

\title{Counterexamples expose gaps in the proof of time complexity for cover trees introduced in 2006}

\author{Yury Elkin\thanks{e-mail: yura.elkin@gmail.com}\\ %
        \scriptsize 
        Department of Computer Science, University of Liverpool %
\and Vitaliy Kurlin\thanks{e-mail: vitaliy.kurlin@gmail.com}\\ %
     \scriptsize Department of Computer Science, University of Liverpool.} 

\teaser{
  \centering
	\begin{tikzpicture}[align=center, node distance = 1.0cm, scale = 0.45]

    \node (scaleinf) {$l = \infty$};
    \node [below = 12.5pt of scaleinf](scalemidinf) {};
     \node[below=25pt of scaleinf] (scale6) {$l = 5$};
    \node[below of =scale6] (scale5) {$l = 4$};
    \node[below of =scale5] (scale4) {$l = 3$};
	\node[below of =scale4] (scale3) {$l = 2$};
	\node[below of =scale3] (scale2) {$l = 1$};
	\node[below=25pt of scale2] (scale1) {$l = -\infty$};

 	\node [blockzm2,  right=70pt of scaleinf ] (node1x) {$r$};
 	\node [blockzm2, below=25pt of node1x ] (node1a) {$r$};
 	\node [blockzm2, below of = node1a ] (node1b) {$r$};
 	\node [blockzm2,  below of = node1b ] (node1c) {$r$};
 	\node [blockzm2,   below of = node1c ] (node1d) {$r$};
 	\node [blockzm2,  below of = node1d ] (node1e) {$r$};
 	\node [blockzm2,   below =25pt of node1e ] (node1f) {$r$};
 	
 	\draw[dashed, ->] (node1x) -> (node1a);
 	    \draw[->] (node1a) -> (node1b);
 	     \draw[->] (node1b) -> (node1c);
 	      \draw[->] (node1c) -> (node1d);
 	       \draw[->] (node1d) -> (node1e);
 	      \draw[dashed, ->] (node1e) -> (node1f);

 		\node [blockzm2, right=10pt of node1b ] (node2b) {$p_{4}$};
 	\node [blockzm2,  below of = node2b ] (node2c) {$p_{4}$};
 	\node [blockzm2,   below of = node2c ] (node2d) {$p_{4}$};
 	\node [blockzm2,  below of = node2d ] (node2e) {$p_{4}$};
 	\node [blockzm2,   below =25pt of node2e ] (node2f) {$p_{4}$};
 	
 	    \draw[->] (node1a) -> (node2b);
 	    
 	    \draw[->] (node2b) -> (node2c);
 	      \draw[->] (node2c) -> (node2d);
 	       \draw[->] (node2d) -> (node2e);
 	      \draw[dashed, ->] (node2e) -> (node2f);

 		\node [blockzm2, right = 10pt of node2c ] (node3c) {$p_{3}$};
 	\node [blockzm2,   below of = node3c ] (node3d) {$p_{3}$};
 	\node [blockzm2,  below of = node3d ] (node3e) {$p_{3}$};
 	\node [blockzm2,   below =25pt of node3e ] (node3f) {$p_{3}$};
 	
 	        \draw[->] (node2b) -> (node3c);
 	      \draw[->] (node3c) -> (node3d);
 	       \draw[->] (node3d) -> (node3e);
 	      \draw[dashed, ->] (node3e) -> (node3f);

 		\node [blockzm2,   left = 10 pt of  node1d ] (node4d) {$p_{2}$};
 	\node [blockzm2,  below of = node4d ] (node4e) {$p_{2}$};
 	\node [blockzm2,   below =25pt of node4e ] (node4f) {$p_{2}$};

            \draw[->] (node1c)   -> (node4d);
 	       \draw[->] (node4d) -> (node4e);
 	      \draw[dashed, ->] (node4e) -> (node4f);

 	\node [blockzm2,  left = 10 pt of node4e ] (node5e) {$p_{1}$};
 	\node [blockzm2,   below =25pt of node5e ] (node5f) {$p_{1}$};
 	
 	\draw[->] (node4d) -> (node5e);
 	 \draw[dashed, ->] (node5e) -> (node5f);
 	
 	\node[rectangle, draw, fill=white, 
    text width=1.0em, text centered, rounded corners, minimum height=50pt, minimum width = 1.0em, right = 220pt of scalemidinf] (ynode1a) {$r$};
    
    	\node[rectangle, draw, fill=white, 
    text width=1.0em, text centered, rounded corners, minimum height=50pt, minimum width = 1.0em, below = 10pt of ynode1a] (ynode1b) {$r$};
    
    	\node[rectangle, draw, fill=white, 
    text width=1.0em, text centered, rounded corners, minimum height=80pt, minimum width = 1.0em, below = 10pt of ynode1b] (ynode1c) {$r$};

    \node[blockzm2, above right = -18 pt and 10 pt of ynode1b] (ynode2a) {$p_4$};
    
    	\node[rectangle, draw, fill=white, 
    text width=1.0em, text centered, rounded corners, minimum height=110pt, minimum width = 1.0em, below = 10pt of ynode2a] (ynode2b) {$p_4$};
    
    \node[rectangle, draw, fill=white, 
    text width=1.0em, text centered, rounded corners, minimum height=110pt, minimum width = 1.0em, right = 10pt of ynode2b] (ynode3a) {$p_3$};

        \node[blockzm2, above left = -15 pt and 10 pt of ynode1c] (ynode4a) {$p_2$};
    
    	\node[rectangle, draw, fill=white, 
    text width=1.0em, text centered, rounded corners, minimum height=55pt, minimum width = 1.0em, below = 10pt of ynode4a] (ynode4b) {$p_2$};
    
    \node[rectangle, draw, fill=white, 
    text width=1.0em, text centered, rounded corners, minimum height=55pt, minimum width = 1.0em, left = 10pt of ynode4b] (ynode5a) {$p_1$};

    \draw[->] (ynode1a) -> (ynode1b);
    \draw[->] (ynode1b) -> (ynode1c);
    \draw[->] (ynode1a) -> (ynode2a);
    \draw[->] (ynode2a) -| (ynode3a);
    \draw[->] (ynode2a) -> (ynode2b);
    
      \draw[->] (ynode1b) -> (ynode4a);
     \draw[->] (ynode4a) -| (ynode5a);
     \draw[->] (ynode4a) -> (ynode4b);
    
 	\node[blockzm2, right = 285pt of scale6] (xnode1) {$r$};
	\node[blockzm2, below right = 10 pt and 10 pt of xnode1] (xnode2) {$p_4$};
 	\node[blockzm2, below right = 70 pt and 10 pt of xnode1] (xnode4) {$p_2$};
 	\node[blockzm2, below right = 10 pt and 8 pt of xnode4] (xnode5) {$p_1$};
 	\node[blockzm2, below right = 12 pt and 8 pt of xnode2] (xnode3) {$p_3$};
 	
 	\draw[->]	(xnode1) -> (xnode2);
 	\draw[->]	(xnode2) -> (xnode3);
 	\draw[->]	(xnode1) -> (xnode4);
 	\draw[->]	(xnode4) -> (xnode5);

\end{tikzpicture}
	\caption{A comparison of different cover trees built on datasets from Example \ref{exa:implicitexplicitexample}. \textbf{Left:} an implicit cover tree introduced in 2006 contains infinite repetitions of given points, see Definition~\ref{dfn:cover_tree_implicit}. \textbf{Middle:} an explicit cover tree still includes repeated points, see Definition~\ref{dfn:cover_tree_explicit}. \textbf{Right:} a new compressed cover tree is much smaller and includes each point only once, see \cite[Definition~3.5]{2205.10194}.  }
	\label{fig:tripleexample}
}

\abstract{
This paper is motivated by the $k$-nearest neighbors search:
given an arbitrary metric space, and its finite subsets (a reference set $R$ and a query set $Q$), design a fast algorithm to find all $k$-nearest neighbors in $R$ for every point $q\in Q$.
In 2006, Beygelzimer, Kakade, and Langford introduced cover trees to justify a near-linear time complexity for the neighbor search in the sizes of $Q,R$. 
\medskip

\noindent
Section~5.3 of Curtin's PhD (2015) pointed out that the proof of this result was wrong. 
The key step in the original proof attempted to show that the number of iterations can be estimated by multiplying the length of the longest root-to-leaf path in a cover tree by a constant factor.
However, this estimate can miss many potential nodes in several branches of a cover tree, that should be considered during the neighbor search.
The same argument was unfortunately repeated in several subsequent papers using cover trees from 2006.
\medskip

\noindent
This paper explicitly constructs challenging datasets that provide counterexamples to the past proofs of time complexity for the cover tree construction, the $k$-nearest neighbor search presented at ICML 2006, and the dual-tree search algorithm published in NIPS 2009.
\medskip

\noindent
The corrected near-linear time complexities with extra parameters are proved in another forthcoming paper by using a new compressed cover tree simplifying the original tree structure.
} 

\begin{document}



\maketitle


\section{Introduction: neighbor problem and past work}

\begin{table*}
	\centering
	\caption{Results for building structures in terms of the expansion constant $c(R)$ in Definition \ref{dfn:expansion_constant} or KR-type constant $2^{\text{dim}_{KR}}$ in \cite[Section~2.1]{krauthgamer2004navigating}}
	\label{table:KR:construction}
	\begin{tabular}{|V{4.0cm}|V{5.5cm}|V{2.5cm}|V{4.0cm}|}
		\hline
		Data structure, reference     & time complexity   & space & proofs \\
		\hline
		Navigating nets \cite{krauthgamer2004navigating} & $O\big(2^{O(\text{dim}_{KR})} \cdot |R| \log(|R|) \log(\log(|R|) )\big)$, \cite[Theorem~2.6]{krauthgamer2004navigating} & $O(2^{O(\text{dim})}|R|)$ & Not available \\
		\hline
		Cover tree \cite{beygelzimer2006cover}   &  $O(c(R)^{O(1)} \cdot |R| \cdot \log(|R|))$, \cite[Theorem~6]{beygelzimer2006cover} & $O(|R|)$ & Counterexample \ref{cexa:construction_algorithm_of_original_cover_tree} shows that the past proof is incorrect \\ 
		\hline
		Compressed cover tree \cite{2205.10194} &     $O\big(c(R)^{O(1)} \cdot |R| \cdot \log(R) \big)$  &  $O(|R|)$    & \cite[Theorem~3.52]{2205.10194}. \\
		\hline       
	\end{tabular}
\end{table*}

\begin{table*}
	\centering
	\caption{Results for finding exact all $k$-nearest neighbors of one query point $q \in Q$ in terms of the expansion constant $c(R)$ in Definition \ref{dfn:expansion_constant} or KR-type constant $2^{\text{dim}_{KR}}$ in \cite[Section~2.1]{krauthgamer2004navigating}, assuming that all data structures are already built.}
	\label{table:KR:knearest}
	\begin{tabular}{|V{4.0cm}|V{5cm}|V{3cm}|V{4.0cm}|}
		\hline
		Data structure, reference     & time complexity   & space  & proofs \\
		\hline
		Navigating nets \cite{krauthgamer2004navigating} & $O\big(2^{O(\text{dim}_{KR})}(\log(|R|) + k)\big)$ for $k\geq 1$ \cite[Theorem~2.7]{krauthgamer2004navigating} & $O(2^{O(\text{dim})} \cdot |R|)$ & Not available \\
		\hline
		Cover tree \cite{beygelzimer2006cover}   &  $O\big(c(R)^{O(1)}\log(|R|)\big)$ for $k=1$, \cite[Theorem~5]{beygelzimer2006cover} & $O(|R|)$ & Counterexample \ref{cexa:original_all_nearest_neighbors_algorithm} shows that the past proof is incorrect \\ 
		\hline
		Compressed cover tree \cite{2205.10194} &     $O\big(c(R)^{O(1)} \cdot \log(k) \cdot (\log(|R|) + k)\big)$  & $O(|R|)$     & \cite[Theorem~3.84]{2205.10194} \\
		\hline   
		Dual cover tree \cite{ram2009linear}   &  $O\big(c(R)^{O(1)} \cdot c(Q)^{O(1)} \cdot |R|)\big)$ for large set $Q$, \cite[Theorem~3.1]{ram2009linear} & $O(|R|)$ & Counterexample \ref{cexa:dualtreeproof} shows that the past proof is incorrect \\ 
		\hline 
	\end{tabular}
\end{table*}

The search for nearest neighbors was one of the first data-driven problems and led to the neighbor rule for classification 
\cite{cover1967nearest}.
\medskip

\noindent
In a modern formulation, the problem is to find all nearest neighbors in a reference set $R$ for all points from a query set $Q$.
Both sets live in an ambient space $X$ with a distance $d$ satisfying all metric axioms.
The simplest example is $X=\R^n$ with the Euclidean metric, where a query set $Q$ can be a single point or a subset of a larger set~$R$.

\begin{prob}[all nearest neighbors search]
	\label{pro:knn}
Let $Q,R$ be finite subsets of query and reference points in a metric space $(X,d)$.
For any integer $k\geq 1$ and every $q\in Q$, design an algorithm to exactly find a point $p$ minimizing $d(q,p)$ so that the parametrized worst-case time complexity is near-linear in $\max\{|Q|,|R|\}$, where hidden constants may depend on structures of $Q,R$ but not on their sizes $|Q|,|R|$. 
\bs
\end{prob}

\noindent
\textbf{Spacial data structures.}
It is well known that the time complexity of a brute-force approach of finding all 1st nearest neighbors of points from $Q$ within $R$ is proportional to the product $|Q|\cdot|R|$ of the sizes of $Q, R$.
Already by the mid of 1970s real data was big enough to motivate faster algorithms and sophisticated data structures.
\medskip

\noindent
One of the first spacial data structure, a \emph{quadtree} \cite{finkel1974quad}, hierarchically indexes a reference set $R\subset\R^2$ by subdividing its bounding box (a root) into four smaller boxes (children), which are recursively subdivided until final boxes (leaf nodes) contain only a small number of reference points.
A generalization of the quadtree to $\R^n$ exposes an exponential dependence of its computational complexity on $n$, because the $n$-dimensional box is subdivided into $2^n$ smaller boxes.
\medskip

\noindent
The first attempt to overcome this dimensionality curse was the $kd$-tree \cite{bentley1975multidimensional} that subdivides a subset of the reference set $R$ at every  recursion step into two subsets instead of $2^n$ subsets.
\medskip

\noindent
Then more advanced algorithms utilizing spatial data structures have positively impacted various related research areas such as a minimum spanning tree \cite{bentley1978fast}, range search \cite{pelleg1999accelerating},  $k$-means clustering  \cite{pelleg1999accelerating}, and ray tracing \cite{fussell1988fast}.
The spacial data structures for finding nearest neighbors in the chronological order are $k$-means tree \cite{fukunaga1975branch}, $R$ tree \cite{beckmann1990r}, ball tree \cite{omohundro1989five}, $R^*$ tree \cite{beckmann1990r}, vantage-point tree \cite{yianilos1993data}, TV trees \cite{lin1994tv}, X trees \cite{berchtold1996x}, principal axis tree \cite{mcnames2001fast}, spill tree \cite{liu2004investigation}, cover tree \cite{beygelzimer2006cover}, cosine tree \cite{holmes2008quic}, max-margin tree \cite{ram2012nearest}, cone tree \cite{ram2012maximum}.
\medskip

\noindent
\textbf{Expansion constant}. 
The past work starting from \cite{karger2002finding} expressed the time complexities of neighbor search in terms of a dimensionality constant for a finite metric space $X$.
This constant was denoted by $2^{\text{dim}_{KR}}$ in \cite[Section~2.1]{krauthgamer2004navigating} and by $c$ in
\cite[Section~1]{beygelzimer2006cover}.
In any metric space $X$, let $\bar B(p,t)\subseteq X$ be the closed ball with a center $p$ and a radius $t$.
Let $|\bar B(p,t)|$ be the number (if finite) of points in $\bar B(p,t)$. 

\begin{dfn}[expansion constant $c(R)$, {\cite[Definition~3.4]{2205.10194}}]
	\label{dfn:expansion_constant}
	Let $R$ be a finite set in a metric space $X$. 
	The \emph{expansion constant} $c(R)$ is the smallest $c(R)\geq 2$ such that $|\bar{B}(p,2t)|\leq c(R) \cdot |\bar{B}(p,t)|$ for any point $p\in R$ and radius $t\geq 0$. 
\end{dfn}

\noindent
Typically, uniformly distributed datasets have small expansion constants. Using arguments of \cite[Section~4.3]{2205.10194} it can be shown that if $R$ is a uniformly distributed point cloud of $\R^{m}$ we have $c(R) = 2^{m}$. 
However, if a dataset contains even a single outlier, say $R = \{1,2,3,\dots,m,2m\}$, then $c(R) = |R|$. 
\medskip

\noindent
The data structures described below were designed to justify a near-linear time complexity for finding $k$-nearest neighbors. 
\medskip

\noindent
\textbf{Navigating nets}. In 2004 a new data structure was introduced that was a sequence of
progressively finer $\epsilon$-nets on the dataset $R$. In \cite[Theorem~2.7]{krauthgamer2004navigating}
it was claimed that all $k$-nearest neighbors of a query point $q$ are found by navigating nets in time $2^{O(\text{dim}_{KR}(R \cup \{q\})}(k + \log|R|)$, where $\text{dim}_{KR}(R \cup \{q\})$ is an expansion rate of \cite[Section~1.2]{krauthgamer2004navigating}.
All proofs and pseudocodes were omitted. The authors did not reply to our request for details. 
\medskip

\noindent
\textbf{Modified navigating nets} \cite{cole2006searching} were used in 2006 to claim the worst-case time complexity $O(\log(|R|) + (1/\epsilon)^{O(1)})$ for finding the first $(1+\epsilon)$-approximate neighbor parameterized by a constant that depends on a doubling dimension of the ambient space. 
However, only sketch of proof of this result was given. 
\medskip

\noindent
\textbf{Cover trees}. 
In 2006, \cite{beygelzimer2006cover} introduced a cover tree inspired by the navigating nets \cite{krauthgamer2004navigating}. 
This cover tree was designed to prove a worst-case time complexity in the size $|R|$ and the expansion constant $c$ from Definition \ref{dfn:expansion_constant}. 
In particular, \cite[Theorem~5]{beygelzimer2006cover} claimed that cover trees help solve Problem \ref{pro:knn} for $k = 1$ could be solved in $O(c^{12}\log|R|)$ time.

\medskip
\noindent 
\textbf{Further studies in cover trees.} A noteworthy paper on cover trees \cite{kollar2006fast} introduced a new probabilistic algorithm for the nearest neighbor search, as well as corrected the pseudo-code of the cover tree construction algorithm of \cite[Algorithm~2]{beygelzimer2006cover}.
Later in 2015 and 2022, new, more efficient implementations of cover tree were introduced in \cite{izbicki2015faster} and \cite{Sheehy2022}. However, no new time-complexity results were proven. A study \cite{jahanseir2016transforming} explored connections between modified navigating nets \cite{cole2006searching} and cover trees \cite{beygelzimer2006cover}.
 Multiple papers \cite{beygelzimer2006coverExtend}, \cite{ram2009linear}, \cite{curtin2015plug} studied possibility of solving Problem \ref{pro:knn} by using cover tree on both, the query set and the reference set, for further details see Section~\ref{sec:challenges_paired_tree}.

\medskip

\noindent
\textbf{Past challenges}. 
In 2015, Curtin's PhD \cite[section~5.3]{curtin2015improving} pointed out that the proof of \cite[Theorem~5]{beygelzimer2006cover} had a mistake.  It was incorrectly claimed that the number of performed iterations of the nearest neighbors algorithm \cite[Algorithm~1]{beygelzimer2006cover} can be bounded by multiplying the depth of cover tree by some constant factor. 
This claim is false because many potential branches at different levels of a cover tree can be missed. 
The similar mistake was repeated in proof of time complexity method of Insert() method \cite[Theorem~6]{beygelzimer2006cover}, as well as in several subsequent papers: for a dual-tree based all-nearest neighbor search \cite[Theorem~3.1]{ram2009linear}, for a Minimum Spanning Tree \cite[Theorem~5.1]{march2010fast}, for a fast exact max-kernel search \cite[Lemma~5.2]{curtin2013fast}. 
\medskip

\noindent
\textbf{Counterexamples}.
To confirm the discovery of Ryan Curtin \cite[section~5.3]{curtin2015improving}, Example~\ref{exa:tall_imbalanced_tree} will describe a finite metric space $(R,d)$ and its cover tree $\T(R)$, for which the maximal root-to-node path is bounded by $O(\sqrt{|R|})$, but that forces both Algorithm~1 and Algorithm~2 of \cite{beygelzimer2006cover} iterate over all $|R|$ levels of the cover tree $\T(R)$. 
The contradiction will follow by noting that the number of iterations in one particular example has a lower bound $|R|$ despite the claimed upper bound $O(\sqrt{|R|})$ for all datasets. 
\medskip

Here is the summary of the found counterexamples:
\begin{itemize}
\item Counterexample~\ref{cexa:construction_algorithm_of_original_cover_tree} to the proof of \cite[Theorem~2]{beygelzimer2006cover}, 
\item Counterexample~\ref{cexa:original_all_nearest_neighbors_algorithm} to the proof of \cite[Theorem~1]{beygelzimer2006cover},
\item Counterexample \ref{cexa:dualtreeproof} to the proof of \cite[Theorem~3.1]{ram2009linear}.
\end{itemize}

\noindent
Counterexamples for the time complexity of a minimum spanning tree \cite[Theorem~5.1]{march2010fast} can be found in \cite[Section~4.2]{2205.10194}.
\medskip

\noindent
\textbf{New results and compressed cover trees}.
All issues of past approaches are resolved in \cite[Chapter~3]{2205.10194} by defining a new data structure, a compressed cover tree \cite[Definition~3.5]{2205.10194}, which combines the explicit and implicit cover tree into a single simpler structure.
Near-linear time algorithms with new parameters for building a  compressed cover tree and for finding $k$-nearest neighbors are described in \cite[Algorithm~3.5.3]{2205.10194} and \cite[Algorithm~3.7.2]{2205.10194}. 
\medskip

\noindent
To overcome the past issues, we estimate the number of iterations in
\cite[Algorithm~3.5.4]{2205.10194} and \cite[Algorithm~3.7.2]{2205.10194} in 
\cite[Lemma~3.5.9]{2205.10194} and \cite[Lemma~3.7.13]{2205.10194} , respectively.
In \cite[Corollary~3.5.11]{2205.10194} it is shown that a compressed cover tree can be constructed in time $O(c(R)^{10}\log(|R|)|R|)$ and \cite[Theorem~3.7.14]{2205.10194} shows that $k$ nearest neighbors of any point $q$ can be found in time  $$O(c(R \cup \{q\})^3 \cdot \log_2(k) \cdot (c(R \cup \{q\})^9 \cdot \log_2(|R|) + k)).$$
Tables \ref{table:KR:construction} and \ref{table:KR:knearest} summarize all known cover tree methods and their contributions for $k$-nearest neighborhood search into two tables. 
\medskip

\section{Original cover trees introduced in 2006}
\label{sec:original_cover_trees}

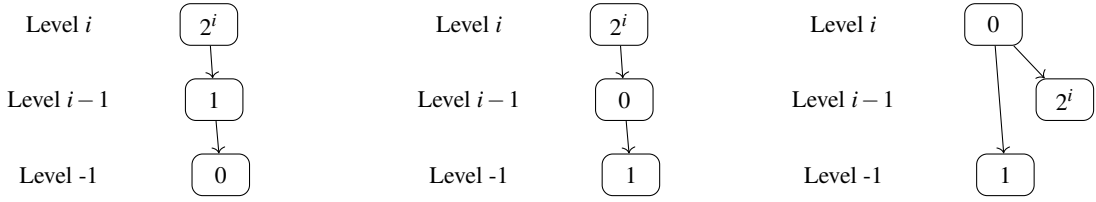
\begin{figure*}
\centering
   \begin{subfigure}{.30\textwidth}
  \centering
  \begin{tikzpicture}[align=center, node distance = 1.0cm, scale = 0.45]

	\node (scalem2text) {Level $i$};
	\node[below of =scalem2text] (scalem1text) {Level $i-1$};
	\node[below of =scalem1text] (scalem0text) {Level -1};
	\node [blockz,  right=30pt of scalem2text ] (node1) {$2^{i}$};
    \node [blockz,  right=25pt of scalem1text ] (node2) {1};
	\node [blockz,  right=32pt of scalem0text ] (node3) {0};
	

	  \draw[->] (node1) -> (node2);
	   \draw[->] (node2) -> (node3);

\end{tikzpicture}
  \label{fig:cover_tree_variant_one}

\end{subfigure}
\begin{subfigure}{.30\textwidth}
  \label{fig:cover_tree_variant_two}
  \centering
  \begin{tikzpicture}[align=center, node distance = 1.0cm, scale = 0.45]

	\node (scalem2text) {Level $i$};
	\node[below of =scalem2text] (scalem1text) {Level $i-1$};
	\node[below of =scalem1text] (scalem0text) {Level -1};
	\node [blockz,  right=30pt of scalem2text ] (node1) {$2^{i}$};
    \node [blockz,  right=25pt of scalem1text ] (node2) {0};
	\node [blockz,  right=32pt of scalem0text ] (node3) {1};
	

	  \draw[->] (node1) -> (node2);
	   \draw[->] (node2) -> (node3);

\end{tikzpicture}
\end{subfigure}
\begin{subfigure}{.30\textwidth}
  \centering
  \begin{tikzpicture}[align=center, node distance = 1.0cm, scale = 0.45]

	\node (scalem2text) {Level $i$};
	\node[below of =scalem2text] (scalem1text) {Level $i-1$};
	\node[below of =scalem1text] (scalem0text) {Level -1};
	\node [blockz,  right=30pt of scalem2text ] (node1) {0};
    \node [blockz,  right=50pt of scalem1text ] (node2) {$2^{i}$};
	\node [blockz,  right=32pt of scalem0text ] (node3) {1};
	

	  \draw[->] (node1) -> (node2);
	   \draw[->] (node1) -> (node3);

\end{tikzpicture}
  \label{fig:cover_tree_variant_three}
\end{subfigure}
\caption{
Compressed representations of three different implicit cover trees $\I(R)$ built on same set $R = \{0,1,2^{i}\}$.
See Definition \ref{dfn:cover_tree_implicit}.
}
\label{fig:cover_tree_easy_example}
\end{figure*}

To resolve Problem \ref{pro:knn} effectively \cite{beygelzimer2006cover} introduced a new data structure, cover tree, the idea of which was to encode data of the reference set $R$ into a leveled tree. Using this tree, a new algorithm \cite[Algorithm~1]{beygelzimer2006cover} was introduced, which was used to find the nearest neighbor of a given query point $q$. The idea was to travel from the root node of the tree, located on the highest level towards the leaf nodes on the lowest level, memorizing the current best candidate for the nearest neighbor and eliminating the branches, which were clearly too far from the query point. Compared to the brute-force search, the benefit of this procedure is that we avoid computing the distance of a query point to a large number of points which are eliminated in large batches during the search.

\medskip

\noindent
Implicit and explicit cover trees are visualizations of finite metric spaces, that were discovered in \cite[Section~2]{beygelzimer2006cover}. However, only the definition of the implicit cover tree was formally stated.

\begin{dfn}[Implicit cover tree $\I(R)$, {\cite[Section~2]{beygelzimer2006cover}}]
	\label{dfn:cover_tree_implicit}
	Let $R$ be a finite set in a metric space $(X,d)$. 
	\emph{An implicit cover tree} $\I(R)$ is a tree on a subset of $R \times \Z \cup \{-\infty, +\infty\}$ with a root $r \in R$ and a \emph{level} function $l : R \rightarrow \Z$ satisfying the conditions below.
	\medskip
	
	\noindent
	(\ref{dfn:cover_tree_implicit}a)
	\emph{Root condition} :  
	The level of the root node $r$ is $l(r) = \infty$. 
	\medskip
	
	\noindent
	(\ref{dfn:cover_tree_implicit}b)
	\emph{Node condition} :  
	For all points $p \in R$ and for all indices $i \in (-\infty, l(p)+1)$
	there exists a node $(p,i)$ in the tree $\I(R)$.
	\medskip
	
	\noindent
	(\ref{dfn:cover_tree_implicit}c)
	\emph{Covering condition} : 
	for every node $(q,i) \in \I(R)$ there exists a \emph{parent} $(p,i+1) \in \I(R)$ 
     such that $d(q,p) \leq 2^{i+1}$, this parent node $p$ has a single link to its  \emph{child} node $q$ in the tree $\I(R)$. 
	\medskip
	
	\noindent
	(\ref{dfn:cover_tree_implicit}d)
	\emph{Separation condition} : 
	for $i \in \Z$ and the \emph{cover set} 
	$C_i = \{p \in R \mid l(p) \geq i\}$,
	the minimum inter-point distance $d_{\min}(C_i) = \min\limits_{p \in C_{i}}\min\limits_{q \in C_{i}\setminus \{p\}} d(p,q)$ is larger than $2^{i}$.
	\medskip
	
	\noindent
 	For any node $p\in\I(R)$, $\Child(p,i)$ denotes the set consisting of all children of the node $(p,i)$, including the node $(p,i-1)$ on the level below.
	For any node $p \in\I(R)$, define the \emph{node-to-root} path as a unique sequence of nodes $w_0,\dots,w_m$ such that $w_0 = p$, $w_m$ is the root and $w_{j+1}$ is the parent of $w_{j}$ for $j=0,...,m-1$. 
	A node $q \in\I(R)$ is a \emph{descendant} of a node $p$ if $p$ is in the node-to-root path of $q$. 
	A node $p$ is an \emph{ancestor} of $q$ if $q$ is in the node-to-root path of $p$. 
	Let $\Desc(p,i)$ be the set of all descendants of node $(p,i)$, including $(p,i-1)$. 
	\bs
\end{dfn}

\noindent
The explicit cover tree is obtained from an implicit cover tree by collapsing into a single node all nodes from any infinite non-branched path $(p,i)\to(p,i-1)\to\cdots \to(p,-\infty)$, see the left and middle pictures of Fig.~\ref{fig:tripleexample}, as formalized below.

\begin{dfn}[Explicit cover tree $\T(R)$, {\cite[Section~2]{beygelzimer2006cover}}]
\label{dfn:cover_tree_explicit}
Let $R$ be a finite set in a metric space $(X,d)$. Let $\I(R)$ be implicit cover tree of Definition \ref{dfn:cover_tree_implicit}. An \emph{explicit cover tree} $\T(R)$ is a quotient tree $\I(R) / \backsim$, where $(p,i) \backsim (q,j)$, if $p = q$ and $\Child(p,t)$ consist of the nodes $(p, t-1)$ for all $t \in [\min(i,j)+1, \max(i,j)]$. 
\end{dfn}

\noindent
Since nodes containing different points are never glued together, we denote an arbitrary node of explicit cover tree $\T(R)$ by $(p,[i])$, where $p \in R$ is the point stored in the node and $[i]$ is equivalence class of $(p,i)$ in $[\backsim]$. 

\begin{exa}[Three point example]
Let $R = \{0,1,2^{i}\}$ for some large $i \in \Z_{+}$ and let $d(x,y) = |x-y|$ be the Euclidean metric on $\R$. 
There are multiple ways to construct an implicit cover tree $\I(R)$. Assume that $2^{i}$ is chosen to be the root node. Then $\I(R)$ will contain an infinite chain $\{(2^{i},j) \mid j \in \Z\}$, in such a way that for all $j$ node $(2^{i},j) $ is parent of $(2^{i},j-1)$.

\medskip
\noindent
Let us now insert points $\{0,1\}$. Since $d(2^{i},1) = 2^{i}-1$ and $d(2^{i},0) = 2^{i}$, by conditions (\ref{dfn:cover_tree_implicit}b) and (\ref{dfn:cover_tree_implicit}c) either $l(0) = i-1$ or $l(1) = i-1$. Let us choose $l(0) = i-1$, then $\I(R)$ will contain chain $\{(0,j) \mid j \in (-\infty, i-1] \cap \Z\}$ in its vertex set, where $(0,j)$ will be parent of $(0,j-1)$ for all $j \in (-\infty, i-1]$ and $(2^{i},i)$ will be parent of $(0,i-1)$. Since $d(0,1) = 1$ and point $0$ minimizes the distance $d(1, \{0,2^{i}\})$ we have $l(1) = -1$. Therefore $\I(R)$ will contain chain $\{(1,j) \mid j \in (-\infty, -1] \cap \Z\}$ and $(1,-1)$ will be child of $(0,0)$.

\medskip
\noindent
The compressed representation of $\I(R)$ is illustrated in Figure~\ref{fig:cover_tree_easy_example} (middle). Explicit cover tree $\T(R)$ consists of nodes: $(2^{i}, [i])$, $(2^{i}, [i-1])$, $(0,[i-1])$, $(0,[-1])$, $(1,[-1])$,
where 
\begin{itemize}
    \item $(2^{i}, [i])$ has two children $(2^{i}, [i-1])$,  $(0,[i-1])$ on level $i-1$.
    \item $(0,[i-1])$ has two children $(0, [-1])$,  $(1,[-1])$ on level $-1$.
    \item No other children are present. 
\end{itemize}

\end{exa}

\begin{exa}[a short train line tree]
	\label{exa:implicitexplicitexample}
	Let $G$ be the unoriented metric graph consisting of two vertices $r,q$ connected by three different edges $e,h,g$ of lengths $|e| = 2^6$ , $|h| = 2^{3}$ , $|g| = 1$. Let $p_{4}$ be the middle point of the edge $e$. 
	Let $p_{3}$ be the middle point of the subedge $(p_4 , q)$. 
	Let $p_{2}$ be the middle point of the edge $h$.
	Let $p_{1}$ be the middle point of the subedge $(p_{2}, q)$. 
	Let $R = \{p_1, p_2,p_3,p_4,r\}$. 
	We construct an implicit cover tree $\I(R)$ by choosing the level $l(p_i) = i$ and by setting the root to be $r$.
	Then $\I(R)$ satisfies all the conditions of Definition \ref{dfn:cover_tree_implicit}, see a comparison of the three cover trees in Fig.~\ref{fig:tripleexample}.
	\bs
\end{exa}

\section{Challenging datasets for original cover trees}
\label{sec:challenging_datasets}

In this section Example \ref{exa:tall_imbalanced_tree} introduces a dataset $R$ and its cover tree $\I(R)$, which will be used to show that key steps in proofs of time complexity estimates of cover tree construction algorithm \cite[Theorem~5]{beygelzimer2006cover} and the nearest neighbor search algorithm \cite[Theorem~6]{beygelzimer2006cover} are incorrect. Since the same false arguments were later repeated in the papers \cite{march2010fast} and \cite{ram2009linear}, 
we provide a detailed counterexamples in Sections \ref{sec:tree_construction} and \ref{sec:nn_glich} that expose the contradiction within each of the proof of the theorems. 

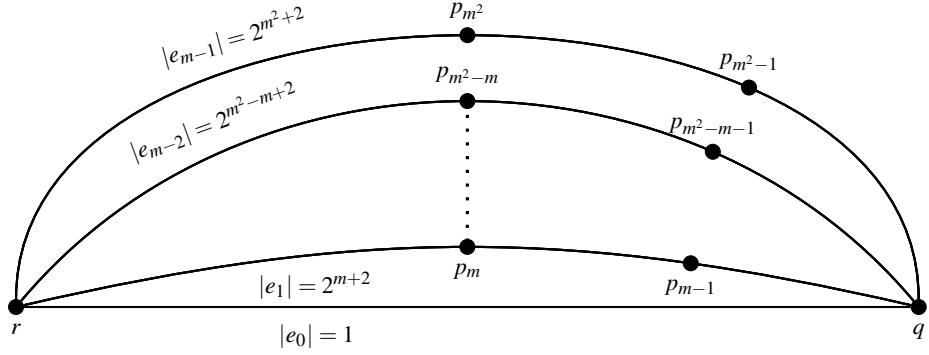
\begin{figure*}
	\centering
	\tikzset{markpos/.style args={#1 at #2}{decoration={
  markings,
  mark=at position #2 with {\coordinate(#1);}},postaction={decorate}}}
  
\begin{tikzpicture}
\node[circle,fill=black,inner sep=2pt,draw,label = below:{$r$}] (a) at (180:6cm) {};
\node[circle,fill=black,inner sep=2pt,draw,label = below:{$q$}] (b) at (0:6cm) {};
\draw[thick] (a) edge[bend left=90, markpos=mymark1 at 0.5] (b);
\draw[thick] (a) edge[bend left=90, markpos=edgemark1 at 0.30] (b);
\draw[thick] (a) edge[bend left=90, markpos=anothermark1 at 0.75] (b);

\draw[thick] (a) edge[bend left=50, markpos=mymark2 at 0.5] (b);
\draw[thick] (a) edge[bend left=50, markpos=edgemark2 at 0.25] (b);
\draw[thick] (a) edge[bend left=50, markpos=anothermark2 at 0.75] (b);

\draw[thick] (a) edge[bend left=13, markpos=mymark3 at 0.5] (b);
\draw[thick] (a) edge[bend left=13, markpos=edgemark3 at 0.33] (b);
\draw[thick] (a) edge[bend left=13, markpos=anothermark3 at 0.75] (b);

\draw[thick] (a) edge[markpos=edgemark4 at 0.33] (b);

\node[circle,fill=black,inner sep=2pt,draw,label = {$p_{m^2}$}] (pm1) at (mymark1) {};
\node[circle,fill=black,inner sep=2pt,draw,label = {$p_{m^2-m}$}] (pm2) at (mymark2) {};
\node[circle,fill=black,inner sep=2pt,draw,label = below:{$p_{m}$}] (p2) at (mymark3) {};

\node[circle,fill=black,inner sep=2pt,draw,label = {$p_{m^2-1}$}] (pm1z) at (anothermark1) {};
\node[circle,fill=black,inner sep=2pt,draw,label = {$p_{m^2-m-1}$}] (pm2z) at (anothermark2) {};
\node[circle,fill=black,inner sep=2pt,draw,label = below:{$p_{m-1}$}] (p2z) at (anothermark3) {};

\node[label = below:{$|e_{0}| = 1$}] (l1) at (edgemark4){};

\node[label = below:{$|e_{1}| = 2^{m + 2} $}] (l2) at (edgemark3){};

\node[label = {[rotate=20]above:$|e_{m-2}| = 2^{m^2 -m + 2} $}] (l2) at (edgemark2){};
\node[label = {[rotate=15]above:$|e_{m-1}| = 2^{m^2+2} $}] (l3) at (edgemark1){};

  \draw[loosely dotted, very thick] (mymark2) -- (mymark3);

\end{tikzpicture}
	\caption{The graph $G$ and the dataset $R$ defined in Example \ref{exa:tall_imbalanced_tree}}
	\label{fig:GraphConstructionOfExample}
\end{figure*}

\begin{exa}[tall imbalanced tree]
	\label{exa:tall_imbalanced_tree}
	
	For any integer $m > 10$, let $G$ be a metric graph pictured in Figure \ref{fig:GraphConstructionOfExample} that has two vertices $r,q$ and $m+1$ edges $(e_i)$ for $i \in \{0, ..., m\}$, and the length of each edge $e_i$ is $|e_i| = 2^{m \cdot i +2}$ for $i \geq 1$.
	Finally, set $|e_0| = 1$. For every $i \in \{1, ..., m^2\}$ if $i$ is divisible by $m$ we set $p_{i}$ be the middle point of $e_{i / m}$ and for every other $i$ we define $p_i$ to be the middle point of segment $(p_{i+1}, q)$. 
	Let $d$ be the induced shortest path metric on the continuous graph $G$.
	Then $d(q,r) = 1$, $d(r, p_{i}) = 2^{i+1} + 1$, $d(q,p_{i}) = 2^{i} $.
	If $i > j $ and $\ceil{\frac{i}{m}} = \ceil{\frac{j}{m}}$, then $$ d(p_{j}, p_{i})  = \sum\limits_{t = j+1}^{i} 2^{t}.$$
	We consider the reference set $R = \{r\} \cup \{p_{i} \mid i=1,2,3,...,m^2 \}$ with the metric $d$.
	
	\medskip
	\noindent
	Let us define an implicit cover tree $\I(R)$ by setting $r$ to be the root node and $l(p_{i}) = i$ for all $i$. 
	For all $i \in 1,...,m^2$:
	If $i$ is divisible by $m$, we set  $(r,i+1)$ to be the parent of $(p_{i},i)$. 
	If $i$ is not divisible by $m$, we set $(p_{i+1},i+1)$ to be the parent of $(p_{i},i)$. 
	For every $i$ divisible by $m$, the point $p_i$ is in the middle of edge $e_{i / m}$, hence $d(p_{i}, r) \leq 2^{i + 1}$.
	For every $i$ not divisible by $m$, by definition, $p_i$ is the middle point of $(p_{i+1},q)$.
	Therefore, we have $d(p_i, p_{i+1}) \leq 2^{i+1}$. 
	Since for any point $p_i$ distance to its parent is at most $2^{i+1}$, the tree $\I(R)$ satisfies covering condition (\ref{dfn:cover_tree_implicit}b).
	For any integer $t$, the cover set is $C_t = \{r\} \cup \{ p_i \mid i \geq t\}$. 
	We will prove that $C_t$ satisfies (\ref{dfn:cover_tree_implicit}c). 
	Let $p_{i} \in C_t$.
	If $i$ is divisible by $m$, then 
	$d(r, p_i) = 2^{i+1} \geq 2^{t + 1} > 2^t.$ 
	If $i$ is not divisible by $m$, then
	$d(r, p_{i}) = d(r,q) + d(q,p_{i})  = 1 + 2^{i+1}  > 2^{t}$. 
	\medskip
    \noindent	
	Then the root $r$ is separated from the other points by the distance $2^t$. 
	Consider arbitrary points $p_{i}$ and $p_{j}$ with indices $i > j \geq t $ and $\ceil{\frac{i}{m}} = \ceil{\frac{j}{m}}$.
	Then
	$$d(p_{i}, p_{j}) = \sum^{i}_{s = j+1} 2^{s}  \geq 2^{j+1} \geq 2^{t+1} > 2^{t}.$$
	On the other hand, if $i > j \geq t $ and $\ceil{\frac{i}{m}} \neq \ceil{\frac{j}{m}}$, then
	$$d(p_{i} , p_{j}) = d(p_{i},q) + d(p_{j} ,q)  \geq  2^{i} + 2^{j} \geq 2^{j+1} \geq 2^{t+1} > 2^t.$$
	For any $t$, we have shown that all pairwise combinations of points of $C_{t}$ satisfy condition~(\ref{dfn:cover_tree_implicit}c).
	Hence this condition holds for the whole tree $\I(R)$.
	\bs
\end{exa}


\noindent
Let us now define the explicit depth, that corresponds to maximal root-to-node path of any cover tree. 
By Definition \ref{dfn:cover_tree_explicit} an explicit cover tree $\T(R)$ is a quotient of $\I(R) / \backsim$, where we collapse all the chains having only a single self-child into a single node. 
Nodes of $\I(R) / \backsim$ are denoted by $(p,[i])$, where $[i]$ is the equivalence class of integer $i$ in $\backsim$. 
By  \cite[Lemma~4.3]{beygelzimer2006cover} the depth of any node $(p,[i])$ is "defined as the number of explicit grandparent nodes on the path from the root to $p$ in the lowest level in which $p$ is explicit". The explicit depth of a node $(p,[i])$ in any explicit cover tree $\T$ is introduced in Definition \ref{dfn:explicit_depth_for_compressed_cover_tree} below using the most natural interpretation of the aforementioned quotes. 

\begin{dfn}[Explicit depth for explicit cover tree]
	\label{dfn:explicit_depth_for_compressed_cover_tree}
	Let $R$ be a finite subset of a metric space with a metric $d$.
	Let $\T(R)$ be an explicit cover tree on $R$. 
	For any $(p,[j]) \in \T(R)$, let $s = (w_0, ... , w_m)$ be a node-to-root path of $(p,[j])$, where
	$w_0 = (r, [+\infty])$ and $w_m = (p,[j])$. We define $D(p, [j])$ to be the number of nodes $|s|$ in the path $s$. The \emph{explicit depth} of a cover tree is defined as the size of maximal node-to-root path $$ D(\T(R)) = \max_{(p,[j]) \in \T(R)}D(p,[j]).$$
	\bs
\end{dfn}

\noindent
Lemma \ref{lem:tall_imbalanced_tree_explicit_depth} shows that the cover tree of Example \ref{exa:tall_imbalanced_tree} the maximal node-to-root has $2 \cdot \sqrt{|R|}$ size, where $|R|$ is the size of dataset. 

\begin{lem}
	\label{lem:tall_imbalanced_tree_explicit_depth}
	Let $\T(R)$ be a compressed cover tree on the set $R$ from Example \ref{exa:tall_imbalanced_tree} for some $m \in \Z$. 
    The explicit depth $D(\T(R))$ of Definition~\ref{dfn:explicit_depth_for_compressed_cover_tree} has the upper bound $2m+1$. 
	\bs
\end{lem}
\begin{proof}
    Note first that root node $r$ contains exactly $m$ non-trivial children.
    Consider arbitrary node $p_i$.
    If $i$ is divisible by $m$, then $r$ is the parent of $p_i$. It follows that we can reach root node $(r,[+\infty])$ in at most $m+1$ steps from $p_i$. 
	
	\medskip
	
	\noindent
	Let us now consider an index $i$ that is not divisible by $m$. 
	Note that  $(p_{j+1}, [j+1])$ is the parent of $(p_j,[j])$ for all $j \in [i, m \cdot \ceil{i / m} - 1]$.  
	Then the path $s$ consisting of all ancestors of $p_i$ from $p_i$ to $p_{m \cdot \ceil{i / m}}$ has the form $((p_{i},[i]), (p_{i+1},[i+1]), ..., (p_{m \cdot \ceil{i / m}}, [m \cdot \ceil{i / m} ] ))$. 
	Note that $|s| \leq m$. Since $ m \cdot \ceil{i / m}$ is divisible by $m$, by the first paragraph the node to root path $l$ from $ (p_{m \cdot \ceil{i / m}}, [m \cdot \ceil{i / m} ] )$ to $(r,[+\infty])$ takes at most $m+1$ steps.
	Therefore $$D(p_i, [i]) \leq |s \cup l| \leq |s| + |l| \leq (m + 1) + m \leq 2m+1,$$
	which proves the claim.
\end{proof}


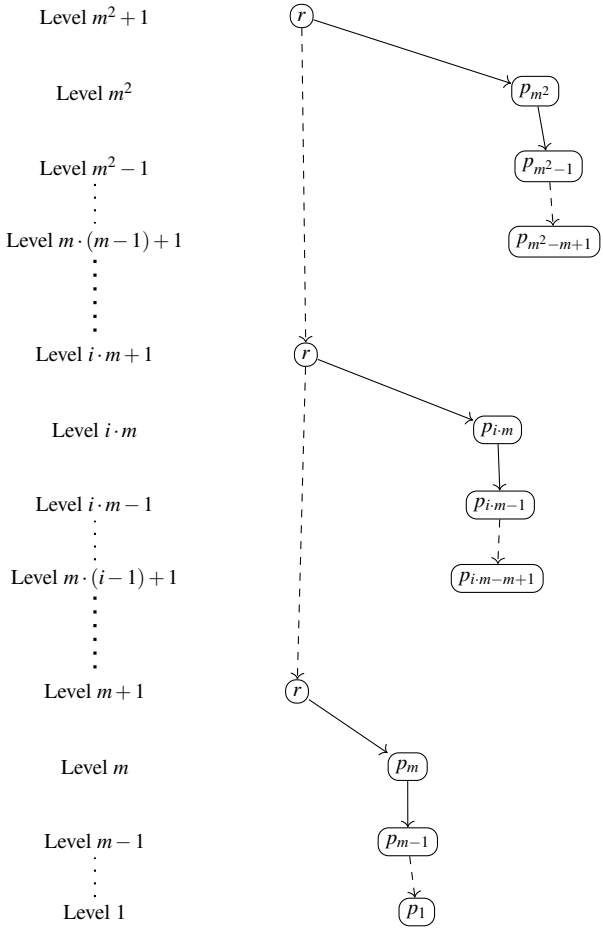
\begin{figure}
	\centering
	
	\begin{tikzpicture}[align=center, node distance = 1.0cm, scale = 0.25]

\scriptsize

	\node (scalem2text) {Level $m^2 + 1$};

	\node[below of =scalem2text] (scalem1text) {Level $m^2 $};
	\node[below of =scalem1text] (scalemh01text) {Level $m^2 - 1$};
	\node[below = 15 pt of scalemh01text] (scalem11text) {Level $m \cdot (m-1) + 1$};

		\node[below = 30 pt of  scalem11text] (scalem1text2) {Level $i \cdot m + 1 $};
		
			\node[below of = scalem1text2] (scalem1text22) {Level $i \cdot m$};
		
	\node[below of =scalem1text22] (scalemh01text2) {Level $i \cdot m - 1$};
	\node[below = 15 pt of scalemh01text2] (scalem11text22) {Level $m \cdot (i-1) + 1$};

		\node [blockzflexi,  right=50pt of scalem2text ] (noderoot) {$r$};

	\node [blockzflexi,  right=140pt of scalem1text] (node1) {$p_{m^2}$};
	\node [blockzflexi, right=135pt of scalemh01text] (node2) {$p_{m^2-1}$};

			\node [blockzflexi, right=120pt of scalem11text] (node4) {$p_{m^2 - m +1}$};

				\node [blockzflexi,  right=124pt of scalem1text22] (node1d) {$p_{i \cdot m}$};
				
					\node [blockzflexi,  right=50pt of scalem1text2] (node1dcopyr) {$r$};

	\node [blockzflexi, right=115pt of scalemh01text2] (node2d) {$p_{i \cdot m-1}$};
			\node [blockzflexi, right=100pt of scalem11text22] (node4d) {$p_{i \cdot m - m +1}$};

	\node[below = 30 pt of scalem11text22] (scalem1textla) {Level  $m+1$};

		\node[below  of =  scalem1textla] (scalem11text2) {Level  $m$};
		
			\node [blockzflexi,  right=50pt of scalem1textla] (noderootcopy3) {$r$};

	\node[below of =scalem11text2] (scalemm1text) {Level $m-1$};
	\node[below = 15 pt of scalemm1text] (level1) {Level $1$};

			\node [blockzflexi, right=95pt of scalem11text2] (node5) {$p_{m}$};
				
			\node [blockzflexi, right=85pt of scalemm1text] (node6) {$p_{m-1}$};
			
			\node [blockzflexi, right=100pt of level1] (node7) {$p_{1}$};


	  \draw[->] (noderoot) -> (node1);
	  \draw[->,dashed] (noderoot) -> (node1dcopyr);
	  \draw[->] (node1dcopyr) -> (node1d);
	  \draw[->] (node1) -> (node2);
	  \draw[->, dashed] (node2) -> (node4);
	  
	   \draw[->] (node1d) -> (node2d);
	  \draw[->, dashed] (node2d) -> (node4d);
	 
	  \draw[->, dashed] (node1dcopyr) -> (noderootcopy3);
	  
	  \draw[->] (noderootcopy3) -> (node5);
	  \draw[->] (node5) -> (node6);
	  \draw[->, dashed] (node6) -> (node7);
	 
	 \draw[loosely dotted, thick] (scalem11text22) -- (scalemh01text2);
	 \draw[loosely dotted, very thick] (scalem11text) -- (scalem1text2);
	  \draw[loosely dotted, thick] (scalemh01text) -- (scalem11text);
	  \draw[loosely dotted, very thick] (scalem11text22) -- (scalem1textla);
	  \draw[loosely dotted, thick] (level1) -- (scalemm1text);

\end{tikzpicture}
	
	\caption{An explicit cover tree built on the dataset $R$ in Example \ref{exa:tall_imbalanced_tree}.
	The node $r$ at the level $i \cdot m+1$ corresponds to the node $(r,[i \cdot m + 1])$ in the explicit cover tree $\T(R)$.
	}
	\label{fig:bad_cover_tree}
\end{figure}

\section{Cover tree construction}
\label{sec:tree_construction}

Counterexample~\ref{cexa:construction_algorithm_of_original_cover_tree} shows that the proof of worst-case time complexity of the Insert() operation for an implicit cover tree \cite[Theorem~6]{beygelzimer2006cover} is incorrect. A correct time complexity for a new compressed cover tree is given in \cite[Corollary~3.53]{2205.10194}. 

\begin{algorithm}
	\caption{Copy-pasted Insert() algorithm for inserting a point $p$ into an implicit cover tree $T$ \cite[Algorithm~2]{beygelzimer2006cover}. This algorithm is launched with $i = l_{\max}$ and $Q_{i} = \{r\}$, where $r$ is the root node of $T$. }
	\label{alg:cover_tree_construction_original}
	\begin{algorithmic}[1]
		\STATE \textbf{Insert}(point $p$, cover set $Q_i$, level $i$)
		\STATE Set $Q = \{\Child(q) \mid q \in Q_i\}$
		\IF {$d(p,Q) > 2^{i}$}
		\STATE \textbf{return} "no parent found"
		\ELSE
		\STATE Set $Q_{i-1} = \{q \in Q \mid d(p,q) \leq 2^{i}\}$
		\IF{\textbf{Insert}$(p,Q_{i-1}, i-1)$ = "no parent found" and $d(p,Q_{i}) \leq 2^{i}$} 
		\STATE Pick $q \in Q_i$ satisfying $d(p,q) \leq 2^{i}$ and insert $p$ into $\Child(q)$, \textbf{return} "parent found"
		\ELSE
		\STATE \textbf{return} "no parent found"
		\ENDIF
		
		\ENDIF 
		
	\end{algorithmic}
\end{algorithm}

\begin{cexa}[for a step in the proof of {\cite[Theorem~6]{beygelzimer2006cover}}]
	\label{cexa:construction_algorithm_of_original_cover_tree}
The idea is based on adding a new point $q$ of Figure \ref{fig:GraphConstructionOfExample} to the tree $\T(R)$ of Example \ref{exa:tall_imbalanced_tree} that lures the Algorithm \ref{alg:cover_tree_construction_original} into using all branches of $\T(R)$. It follows that the Algorithm \ref{alg:cover_tree_construction_original} is launched $O(|R|)$ times. However, in the proof of \cite[Theorem~6]{beygelzimer2006cover} it was claimed that Algorithm \ref{alg:cover_tree_construction_original} is
	launched at most $4 \cdot D(\T(R))$, where $D(\T(R))$ is the explicit depth of explicit cover tree $\T(R)$.
	This is a contradiction, since $D(\T(R)) \leq 2\sqrt{|R|} + 1$ but Algorithm~\ref{alg:cover_tree_construction_original} runs $O(|R|)$ times.
	\medskip
	
    \noindent
	For more details, we cite a part of the proof of \cite[Theorem~6]{beygelzimer2006cover}:
	
	\medskip
	
	"\textbf{Theorem 6} Any insertion or removal takes time at most $O(c^6\log(n))$"
	[In other words the run time of Algorithm \ref{alg:cover_tree_construction_original} is $O(c^6\log(n))$, where $n$ is the number points of original dataset $S$ on which tree $T$ was constructed.]
	
	[\emph{Partial proof:} ]: \emph{" Let $k = c^2 \log(|S|)$ be the maximum explicit depth
	of any point, given by Lemma 4.3. Then the total
	number of cover sets with explicit nodes is at most
	$3k+k = 4 k$, where the first term follows from the fact
	that any node that is not removed must be explicit at
	least once every three iterations, and the additional
	$k$ accounts for a single point that may be implicit for
	many iterations.
	Thus the total amount of work in Steps 1 [Our line 2] and 2 [Our lines 3-5] is
	proportional to $O(k \cdot \max_i|Q_i|)$. Step 3 [Our lines 5-11] requires work
	no greater than step 1 [Our line 2]."}
	
	\medskip
	\noindent
	\textbf{In our interpretation the above arguments says that the total number of times line $1$ [our line 2] was called during the algorithm has the upper bound $4 \cdot D(\T(R))$ , where $D(\T(R))$ is the explicit depth of $\T(R)$, see Definition \ref{dfn:explicit_depth_for_compressed_cover_tree}. In this Counterexample we will show that $\T(R)$ from Example \ref{exa:tall_imbalanced_tree} does not satisfy the claimed inequality.}

\medskip	
\noindent	
	Take the reference set $R$, the compressed cover tree $\T(R)$ and the point $q$ from Example \ref{exa:tall_imbalanced_tree} for any parameter $m > 200$. 
	Assume that we have already constructed tree $\T(R)$. 
	Let us show that $\T(R \cup q)$ constructed by Algorithm~\ref{alg:cover_tree_construction_original} from the input $q, i = m^2+1, Q_i = \{(r,[+\infty])\}$ runs at least $m^2-2$ self-recursions. 
	This will lead to a contradiction since by Lemma \ref{lem:tall_imbalanced_tree_explicit_depth} we have $D(\T(R)) \leq 2m + 1$.
	
	\medskip
	\noindent
	We show by induction on $m$ going down that, for every step $i \in [1, m^2]$, we have $Q_i = \{(r,[i]),(p_i,[i])\}$. 
	The proof for the base case $i = m^2$ is similar to the induction step and thus will be omitted. 
	Assume that $Q_{i}$ has the desired form for some $i$.
	Let us show that the claim holds for $i-1$. 
	For all levels $i-1$ divisible by $m$, the node $(p_{i-1},[i-1]$ is a child of node $(r,[i])$. 
	For all levels $i-1$ not divisible by $m$, the node $(p_{i-1}, [i-1])$ is a child of $p_i$. 
	Since $\T(R)$ contains exactly one node at each level, in both cases we have $Q = \{(r,[i]), (p_{i},[i]), 
	(p_{i-1}, [i-1])\}$. 
	Since $d(q,r) = 1$, $d(q,p_{i}) = 2^{i+1}$ and $d(q,p_{i-1}) = 2^{i}$  we have 
	$$Q_{i-1} = \{p \in Q_i \mid d(p,q) \leq 2^i\} = \{(r , [i-1]),(p_{i-1}, [i-1])\}.$$
	\medskip
	\noindent
	The actual implementation of algorithm \ref{alg:cover_tree_construction_original} iterates over all levels $i$ for which there exists a node in $Q_i$  that contains at least one non-trivial child on level $i-1$ and for which the condition in line $7$ is satisfied. Since for every index $i \in [2,m^2+1]$ we have $Q_i = \{(r,[i]),(p_i,[i])\}$ and since either $r$ or $p_{i}$ has a child at level $i-1$ and the condition in line $7$ is always satisfied, it follows that $m^2-2$ is a low bound for the number $\xi$ of self-recursions. 
	Therefore the contradiction follows from the inequality:
	$$m^2-2 \leq  \xi \leq 4 \cdot D(\T(R)) \leq 8 \cdot (2m + 1) \leq 16 \cdot m + 8$$
	where $m > 20$. 
	\bs

\end{cexa}

\section{Nearest neighbor search}
\label{sec:nn_glich}

Counterexample \ref{cexa:original_all_nearest_neighbors_algorithm} shows that the proof of 
\cite[Theorem~5]{beygelzimer2006cover}, which gives an upper bound for the complexity of Algorithm \ref{alg:cover_tree_k-nearest_original} is incorrect. A correct time complexity estimate for a new $k$-nn algorithm using compressed cover tree is given in \cite[Corollary~3.84]{2205.10194}. 


\begin{algorithm}
	\caption{Copy-pasted \cite[Algorithm~1]{beygelzimer2006cover} based on an implicit cover tree $T$ \cite[Section~2]{beygelzimer2006cover} for nearest neighbor search, which is used in Counterexample \ref{cexa:original_all_nearest_neighbors_algorithm}. The children of a node $q$ of an implicit cover tree are defined as the nodes at one level below $q$ that have $q$ as their parent. In the actual implementation the loop in lines 3-6 runs only for the levels containing nodes with non-trivial children (not coinciding with their parents). 
	}
	\label{alg:cover_tree_k-nearest_original}
	\begin{algorithmic}[1]
		\STATE \textbf{Input} : implicit cover tree $T$, a query point $p$
		\STATE Set $Q_{\infty} = C_{\infty}$ where $C_{\infty}$ is the root level of $T$
		\FOR{$i$ from $\infty$ down to $-\infty$}
		\STATE Set $Q = \{\text{Children}(q) \mid q \in Q_i\}.$
		\STATE Form cover set $Q_{i-1} = \{q \in Q \mid d(p,q) \leq d(p,Q) + 2^{i}\}$
		\ENDFOR
		\STATE \textbf{return} $\text{argmin}_{q \in Q_{-\infty}}d(p,q)$
	\end{algorithmic}
\end{algorithm}

\begin{cexa}[for a step in the proof of {\cite[Theorem~5]{beygelzimer2006cover}}]
	\label{cexa:original_all_nearest_neighbors_algorithm}
	Counterexample \ref{cexa:original_all_nearest_neighbors_algorithm} shows that there is a gap in proof of \cite[Theorem~6]{beygelzimer2006cover}. The counterexample is obtained by running Algorithm \ref{alg:cover_tree_k-nearest_original} for node $q$ of Figure \ref{fig:GraphConstructionOfExample} and tree $\T(R)$ of Example \ref{exa:tall_imbalanced_tree}. 
    It it shown that Algorithm \ref{alg:cover_tree_k-nearest_original} iterates over all branches of $\T(R)$, therefore  lines 3-6 are considered exactly $|R|$ times. However, the proof of \cite[Theorem~6]{beygelzimer2006cover} claimed that the number of times lines 3-6 are considered is bounded by multiplication $\max_i|R_i| \cdot D(\T(R))$, where $D(\T(R))$ is the maximal path-to-root path that has an upper bound $2\sqrt{|R|}$. In this counterexample it will be also shown that $\max_i|R_i| \leq 3$ during the whole iteration of the algorithm, which will lead to contradiction $|R| \leq 3 \cdot 2\sqrt{|R|}$, when $|R|$ is sufficiently big. 
    
    \medskip
    \noindent 
    
   	For more detailed exhibition let us first cite a part of the proof of  \cite[Theorem~5]{beygelzimer2006cover}.
	
	\medskip
	\noindent
	"\textbf{Theorem 5}
	If the dataset $S \cup \{p\}$ has expansion constant $c$, the nearest neighbor of $p$ can be found in time $O(c^{12}\log(n))$."
	
	\smallskip

	\emph{[\emph{Partial proof:}] "Let $Q^{*}$ be the last $Q$ considered by the Algorithm \ref{alg:cover_tree_k-nearest_original} (so $Q^{*}$ consists only of lead nodes with scale $-\infty$). Lemma 4.3 bounds the explicit depth of any node in the tree (and in particular any node in $Q^{*}$) by $k = O(c^2 \log (N))$. Consequently the number of iterations is at most $k|Q^{*}| \leq k \max_i|Q_i|$."}
	
	\medskip
	\noindent
	\textbf{By our interpretation the above argument claims that the total number $\xi$ of times when Algorithm \ref{alg:cover_tree_k-nearest_original} runs lines 3-6 has an upper bound $\xi \leq D(\T(R)) \cdot \max_i|Q_i|.$ Contradiction will be obtained by showing that $\T(R)$ from Example \ref{exa:tall_imbalanced_tree} does not satisfy this inequality.} 

	\medskip
	\noindent	
	Take $R,\T(R)$ and $q$ from Example \ref{exa:tall_imbalanced_tree}. We will apply Algorithm \ref{alg:cover_tree_k-nearest_original} to the tree $\T(R)$ and query point $q$. 
	By Lemma \ref{lem:tall_imbalanced_tree_explicit_depth} the cover tree $\T(R)$ having parameter $m$ has $D(p) \leq 2m+1$ for all $p \in R$. A contradiction to the original argument will follow after showing that $\max|Q_i| \leq 2$ and $\xi \geq m^2 - 2$.
	
	\medskip
	\noindent
	Let us first estimate $\max_i |Q_i|$.  
	Similarly to Counterexample \ref{cexa:construction_algorithm_of_original_cover_tree} we will show that, for every iteration $i \in [1, m^2]$ of lines 3-5 of Algorithm \ref{alg:cover_tree_k-nearest_original}, we have 
	$Q_i = \{(r,[i]),(p_i,[i])\}$. 
	The proof for the basecase $i = m^2$ is similar to the induction step and thus will be omitted. 
	Assume that $Q_{i}$ has the desired form for some $i$.
	Let us show that the claim holds for $i-1$. 
	For all levels $i-1$ divisible by $m$, the node $(p_{i-1}, [i-1]$ is a child of the root $(r,[i])$. 
	For all levels $i-1$ not divisible by $m$, the node $(p_{i-1}, [i-1])$ is a child of $(p_{i}, [i])$. 
	Since $\T(R)$ contains exactly one node at each level, in both cases
	we have $Q = \{(r,[i-1]), (p_{i},[i-1]). (p_{i-1},[i-1])\}$. 
	Since $d(q,r) = 1$, $d(q,p_{i}) = 2^{i+1}$ and $d(q,p_{i-1}) = 2^{i}$, we have 
	$$Q_{i-1} = \{p \in Q_t \mid d(p,q) \leq 2^i + 1\} = \{(r,[i-1]),(p_{i-1}, [i-1])\}$$
	Therefore it follows that $|Q_i| \leq 2$ for all $i \in [1, m^2]$.
	
	\medskip
	\noindent
	The actual implementation of algorithm \ref{alg:cover_tree_k-nearest_original}iterates over all levels $i$ for which there exists a node in $Q_i$ containing at least one non-trivial child at level $i-1$. 
	Since $Q_i = \{r,p_i\}$ and for every index $i \in [2,m^2+1]$, either $(r,[i])$ or $(p_{i},[i])$ has a child on level $i-1$, it follows that  $m^2-2$ is a low bound for the number $\xi$  of iterations. 
	A contradiction follows from  
	$$m^2 - 2 \leq \xi \leq  D(\T(R)) \cdot  \max_i|Q_i| \leq (2m+1) \cdot 2 \leq 4m+2,$$
	for any $m > 20$. 
	
	\bs
\end{cexa}

\section{Challenges of the nearest neighbor search based on paired trees}
\label{sec:challenges_paired_tree}

In 2009 \cite[Theorem~3.1]{ram2009linear} revisited the time complexity for all 1st nearest neighbors and claimed the upper bound $O(c(R)^{12}c(Q)^{4\kappa}\max\{|Q|,|R|\})$, where $c(Q),c(R)$ are expansion constants of the query set $Q$ and reference set $R$.
The degree of bichromaticity $\kappa$ is a parameter of both sets $Q,R$, see \cite[Definition~3.1]{ram2009linear}.
We have found the following issues.

\medskip
\noindent
First, Counterexample~\ref{cexa:dualtreecode} shows that \cite[Algorithm~1]{ram2009linear} for $Q=R$ returns for any query point $q\in Q$ the same point $q$ as its first neighbor. 
Second, Remark~\ref{rem:kappa} explains several possible interpretations of \cite[Definition~3.1]{ram2009linear} for the parameter $\kappa$. 
Third, \cite[Theorem~3.1]{ram2009linear} similarly to \cite[Theorem~5]{beygelzimer2006cover} relied on the same estimate of recursions in the proof of \cite[Lemma~4.3]{beygelzimer2006cover}.
Counterexample~\ref{cexa:dualtreeproof} explains step-by-step why the proof of the time complexity result of \cite[Algorithm~1]{ram2009linear} is incorrect and requires a clearer definition of $\kappa$.

\medskip
\noindent
In 2015 Curtin with the authors above  \cite{curtin2015plug} introduced other parameters:
the imbalance $I_t$ in \cite[Definition~3]{curtin2015plug} and $\theta$ in \cite[Definition~4]{curtin2015plug}.
These parameters measured extra recursions that occurred due to possible imbalances in trees built on $Q,R$, which was missed in the past. 
\cite[Theorem~2]{curtin2015plug} shows that, for constructed cover trees on a query set $Q$ and a reference set $R$, Problem~\ref{pro:knn} for $k=1$ (only 1st nearest neighbors) can be solved in time 
$$O\Big(c^{O(1)}\big(|R| + |Q| + I_t + \theta\big)\Big).\eqno{(*)},$$
where $c$ is expansion constant that depends on $Q$ and $R$.
The problem with this approach is that in worst case $I_t$ is quadratic $O(|R|^2)$. To make the time complexity linear, we would have to show
$I_t = O(c^{O(1)} \cdot \max\{|R| , |Q|\})$. However, no such result exist at the moment. 

\begin{algorithm}
	\caption{
		Copy-pasted \cite[Algorithm~1]{ram2009linear} is analyzed in Counterexamples~\ref{cexa:dualtreecode} and~\ref{cexa:dualtreeproof}.
	}
	\label{alg:cover_tree_k-nearest_dt_original}
	\begin{algorithmic}[1]
		\STATE \textbf{Function} FindAllNN(a node $q_j\in T(Q)$, a subset $R_i$ of a cover set $C_i$ of $T(R)$).
		\IF {$i = -\infty$}
		\STATE for each $q_j \in L(q_j)$ \textbf{return} $\text{argmin}_{r \in R_{-\infty}} d(q,r)$
		\STATE \COMMENT{here $L(q_j)$ is the set of all descendants of the node $q_j$}
		\ELSIF{$j < i$}
		\STATE $\C(R_i) = \{\text{Children}(r) \mid r \in R_i\} $ 
		\STATE $R_{i-1} = \{r \in R \mid d(q_j,r) \leq d(q_j, R) + 2^{i} + 2^{j+2} \}$
		\STATE FindAllNN($q_{j-1}, R_i$) \COMMENT{ $q_{j-1}$ is the same point as $q_j$ on one level below}
		\ELSE{}
		\STATE for each $p_{j-1} \in \text{Children}(q_j)$ FindAllNN($p_{j-1},R_{i}$)
		\ENDIF
	\end{algorithmic}
\end{algorithm}

\medskip
\noindent
The step-by-step execution of Algorithm~\ref{alg:cover_tree_k-nearest_dt_original} will show that the number of reference expansions has a lower bound $O(\max\{|Q|,|R|\}^2)$.  Recall that \cite[End of Section~1]{ram2009linear} defined the all-nearest-neighbor problem as follows. 
"\textbf{All Nearest-neighbors:} For all queries $q \in Q$ find $r^{*}(q) \in R$ such that $r^{*}(q) = \mathrm{argmin}_{r \in R}d(q,r)"$. 
For $Q=R$, the last formula produces trivial self-neighbors.

\medskip
\noindent
In original Algorithm~\ref{alg:cover_tree_k-nearest_dt_original}, the node $q_j$ has a level $j$, a reference subset $R_i\subset R$ is a subset of $C_i$ for an explicit cover tree $\T(R)$. 
The algorithm is called for a pair $q_j, R_{i} = \{(r, [+\infty])\}$, where $q_j$ is the root of the query tree at the maximal level $j = +\infty$,
and $r$ is the root of the reference tree at the maximal level $i = +\infty$. 

\medskip
\noindent
Split Algorithm~\ref{alg:cover_tree_k-nearest_dt_original} into these blocks:
\textbf{lines 2-4} : FinalCandidates,
\textbf{lines 5-9} : reference expansion,
\textbf{lines 9-11} : query expansion.

\begin{cexa}
	\label{cexa:dualtreecode}
	In the notations of  Example \ref{exa:tall_imbalanced_tree}, $m$ is a parameter of $R$.
	Build a compressed cover tree $\T(R)$ as in Figure \ref{fig:bad_cover_tree}. 
	Set $Q = R$. 
	First we show that Algorithm \ref{alg:cover_tree_k-nearest_dt_original} returns the trivial neighbor when $\T(Q)=\T(R)$. 
	\medskip
	\noindent
	We start the simulation with the query node $r$ on the level $m^2+1$, which has the reference subset $R_{m^2+1} = \{(r, [m^2+1])\}$. 
	The query node and the reference set are at the same levels, so we run the query expansions (lines 9-11). 
	The node $r$ has $p_{m^2}$ and $r$ as its children.
	Hence the algorithm goes into the branches that have $p_{m^2}$ as the query node and into the branches that have $r$ as the query node.
	\medskip
	\noindent
	Let us focus on all recursions having $p_{m^2}$ as the query node. In the first recursion involving the node $p_{m^2}$, we have $i = m^2+1, j = m^2$. Thus $j < i$ and we run reference expansions (lines 5-9). 
	The node $(r, [m^2+1])$ has two children at the level $m^2$, so $\C(R_i) = \{(p_{m^2}, [m^2]), (r, [m^2])\}$ . Since
	$d(p_{m^2},p_{m^2}) = 0$ and $d(p_{m^2}, r) = 2^{m^2+1}$ on line 7, we have:
	\begin{align*}
	    R_{m^2}  &= \{r \in \C(R_i) \mid d(q_j,r) \leq 2^{m^2 + 1} + 2^{m^2+2} \} \\
	    &= \{(p_{m^2}, [m^2]), (r, [m^2])\}
	\end{align*}
	Similarly, for $i = m^2, j = m^2-1, q_j = p_{m^2}$, we have 
	$$\C(R_i) =  \{ (p_{m^2}, [m^2]), (p_{m^2-1}, [m^2-1]), (r, [m^2+1])\}$$
	and since $d(p_{m^2 - 1}, p_{m^2}) = 2^{m^2}$ and $d(r,p_{m^2}) = 2^{m^2 + 1}$ we have:
		\begin{align*}
	    R_{m^2-1} &=  \{r \in \C(R_i) \mid d(q_j,r) \leq 2^{m^2} + 2^{m^{2} + 1} \} \\
	    &=  \{(p_{m^2}, [m^2-1]), (p_{m^2-1}, [m^2-1])\}.
	\end{align*}
	For $i = m^2-1, j = m^2-2, q_j = p_{m^2}$, we have 
	$$\C(R_i) =  \{(p_{m^2}, [m^2-2]), (p_{m^2-1}, [m^2-2]), (p_{m^2-2}, [m^2-2])\}.$$
	Since $d(p_{m^2 - 1}, p_{m^2}) = 2^{m^2}$ and $d(p_{m^2-2},p_{m^2}) = 2^{m^2} + 2^{m^2-1}$, we have:
		\begin{align*}
	    R_{m^2-2} &=  \{r \in \C(R_i) \mid d(q_j,r) \leq 2^{m^2 - 1} + 2^{m^{2}} \}   \\
	    &= \{(p_{m^2},[m^2]), (p_{m^2-1}, [m^2-1]), (p_{m^2-2}, [m^2-2])\}.
	\end{align*}
	Finally, for $i = m^2-2, j = m^2-3, q_j = p_{m^2}$, we have $$\C(R_i) =  \{(p_{m^2},[m^2-3]), ... , (p_{m^3-3},[m^2-3])\}$$ and 
	$d(p_{m^2}, p_{m^3-3}) = 2^{m^2} + 2^{m^2-1} + 2^{m^2-2}$. The previous inequalities imply that
	$$R_{m^2-3} =  \{r \in \C(R_i) \mid d(q_j,r) \leq 2^{m^2-2} + 2^{m^{2}-1} \} = \{(p_{m^2}, [m^2-3])\}.$$
	Since  $R_{m^2 -3} = \{p_{m^2}\}$, the nearest neighbor of $p_{m^2}$ will be chosen to be $p_{m^2}$. 
	The same argument can be repeated for all $p_t \in R$. 
	It follows that Algorithm \ref{alg:cover_tree_k-nearest_dt_original} finds trivial nearest neighbor for every point $p_t \in R$. 
	\bs
\end{cexa}







\begin{exa}
	\label{exa:dualtreeprooffaultyconst}
	To avoid the issue of finding trivial nearest neighbors as in Counterexample~\ref{cexa:dualtreecode}, we will modify Example~\ref{exa:tall_imbalanced_tree}. 
	For any integer $m > 100$, let $G$ be a metric graph that has $2$ vertices $r$ and $q$ and $2m-1$ edges $\{e_{0}\} \cup \{e_{1}, ..., e_{m-1},h_{1},..., h_{m-1}\}$.
	The edge-lengths are $|e_i| = 2^{i \cdot m+2}$ and $|h_i| = 2^{i \cdot m+2}$  for all $i \in [1,m]$, finally $|e_{0}| = 1$.
	
	\medskip
	\noindent
	For every $i \in \{1, ..., m^2\}$, if $i$ is divisible by $m$, we set $q_{i}$ to be the middle point of $e_{i / m}$ and $r_{i}$ to be the middle point of $h_{i / m}$. 
	For every other $i$ not divisible by $m$, we define $q_i$ to be the middle point of segment $(q_{i+1}, q)$ and $r_i$ to be the middle point of segment $(r_{i+1}, q)$.
	
	\medskip
	\noindent
	Let $d$ be the shortest path metric on the graph $G$. 
	Then $d(q_i,r) = d(q_i, q) + 1 = 2^{i+1} + 1$ , $d(q_i, r_j) = 2^{i+1} + 2^{j+1}$ and $d(q,r) = 1$. 
	Let $R = \{r, r_{m^2}, r_{m^2-1}, ..., r_1\}$ and let $Q = \{r,  q_{m^2}, q_{m^2-1}, ..., q_1 \}$. 
	Let compressed cover trees $\T(Q),\T(R)$ have the same structure as the compressed cover tree $\T(R)$ in Example~\ref{exa:tall_imbalanced_tree}.
	\bs
\end{exa}

\begin{rem}
	\label{rem:kappa}
	\cite[Definition~3.1]{ram2009linear} introduced the degree of bichromaticity $\kappa$ as follows. 
	\medskip
	\noindent
	"\textbf{Definition 3.1} Let $S$ and $T$ be cover trees built on query set $Q$ and reference set $R$ respectively.
	Consider a dual-tree algorithm with the property that the scales of $S$ and $T$ are kept as close as
	possible – i.e. the tree with the larger scale is always descended. Then, the degree of bichromaticity
	$\kappa$ of the query-reference pair $(Q, R)$ is the maximum number of descends in $S$ between any two
	descends in $T$".
	\medskip
	\noindent
	There are at least two different interpretations of this definition. Our best interpretation is that $\kappa$ is the maximal number of levels in $T$ containing at least one node between any two consecutive levels of $S$. However, if $q$ is a leaf node of $S$, but there are still many levels between level of $q$ and $l_{\min}(T)$, it is not clear from the definition if $\kappa$ includes these levels.
	\medskip
	\noindent
	\cite[page~3284]{curtin2015plug} pointed out that "
	Our results are similar to that of Ram et al. (2009a), but those results depend on a
	quantity called the constant of bichromaticity, denoted $\kappa$, which has unclear relation to
	cover tree imbalance. The dependence on $\kappa$ is given as $c_q^{4\kappa}$ , which is not a good bound,
	especially because $\kappa$ may be much greater than 1 in the bichromatic case (where $S_q = S_r$)".
\end{rem}

\noindent
To keep track of the indices $i,j$ the function call FindAllNN($q_j$, $R_i$) will be expressed as FindAllNN($i,j,q_j, R_i$) in Counterexample \ref{cexa:dualtreecode}.

\begin{cexa}[for a step in the proof of {\cite[Theorem~3.1]{ram2009linear}} ]
	\label{cexa:dualtreeproof}
	We will now show that in addition to the problems in the pseudocode the proof of \cite[Theorem~3.1]{ram2009linear} is incorrect.  Let us consider the following quote from its proof. 
	
	\medskip
	\noindent
	\emph{"\textbf{Theorem 3.1} Given a reference set $R$ of size $N$ and expansion constant $c_R$, a query set $Q$ of size
	$O(N)$ and expansion constant $c_Q$, and bounded degree of bichromaticity $\kappa$ of the $(Q,R)$ pair, the
	FindAllNN subroutine of Algorithm 1 computes the nearest neighbor in $R$ of each point in $Q$ in
	$O(c^{12}_Rc^{4\kappa}_QN)$ time.}
	
	\medskip
	\noindent
	\emph{[\emph{Partial proof:}]
	Since at any level of recursion, the size of $R$ [Corresponding to $\C(R_i)$ in Algorithm \ref{alg:cover_tree_k-nearest_dt_original} ] is bounded by $c_R^4\max_i{R_i}$ (width bound), and the
	maximum depth of any point in the explicit tree is $O(c^2_R \log(N))$ (depth bound), the number of nodes
	encountered in Line 6 is $O(c_R^{6} \max_i |R_i|\log(N))$. Since the traversal down the query tree causes
	duplication, and the duplication of any reference node is upper bounded by $c_Q^{4\kappa}$ , Line 6 [corresponds to line 8 in Algorithm \ref{alg:cover_tree_k-nearest_dt_original}] takes at most
	$c^{4\kappa}_Qc^6_R\max_i|R_i|\log(N)$ in the whole algorithm. }
	"
	
	\medskip
	\noindent
	\textbf{The above arguments claimed the algorithm runs Line 8 at most this number of times:}
	\begin{equation}
		\label{eqa:AuthorClaimingDualTreeKNN}
		\#(\text{Line 8}) \leq D(\T(R)) \cdot \max_i \C(R_i)  \cdot (\text{number of duplications}).
	\end{equation}
	\textbf{It will be shown that cover tree $\T(R)$ from Example \ref{exa:dualtreeprooffaultyconst} does not satisfy the inequality above. }

\medskip	
\noindent	
	Let $X, \T(R), \T(Q), R, Q$ be as in Example \ref{exa:dualtreeprooffaultyconst} for some parameter $m$. We will consider the simulation of Algorithm \ref{alg:cover_tree_k-nearest_dt_original} on pair $(\T(Q),\T(R))$. We note first Lemma \ref{lem:tall_imbalanced_tree_explicit_depth} applied on $\T(R)$ provides $\max_{p \in R}D(p) \leq 2m+1$
As in Counterexample \ref{cexa:original_all_nearest_neighbors_algorithm}, a contradiction will be achieved by showing that $R_i$ and a set of its children $\C(R_i)$ have a constant size bound on any recursion $(i,j)$ of Algorithm \ref{alg:cover_tree_k-nearest_dt_original}.
	
	\medskip
	\noindent
	Since $\T(R)$ contains at most one children on every level $i$ we have $|\C(R_i)| \leq |R_i| + 1$ for any recursion of FindAllNN algorithm. For any $i > m^2$ denote $r_i$ and $q_i$ to be $r$.
	Note first that since $l(q_t) = t$ for any $t \in [1,m^2]$, then $q_t$ is recursed into from FindAllNN($t+1,t+1,p,R_i$), where $p$ is parent node of $q_t$. Therefore it follows that $t \geq i + 1 $ in any stage of the recursion. 
	Let us prove that for any $i \in [1,m^2+1]$ following two claims hold: (1) Function FindAllNN($i$ , $j = i-1$, $q_t$, $R_i$) is called for all $t \geq i + 1$ and (2) We have  $R_i = \{r_{i+1}, r_{i}, r\}$ in this stage of the algorithm. The claim will be proved by induction on $i$. Let us first prove case $i = 2m+1$.  Note that  Algorithm \ref{alg:cover_tree_k-nearest_dt_original} is originally launched from FindAllNN($2m+1,2m+1,r, \{r\}$), 
	therefore the first claim holds. Second claim holds trivially since $r_{2m+2} = r$ and $r_{2m+1} = r$.
\medskip

\noindent	
	Let the claim hold for some $i$, let us show that the claim will always hold for $i-1$. Assume that FindAllNN($i , j = i-1, q_{t}, R_i)$ was called for some $t \geq i+1$. Since $j = i-1$, we perform a reference expansion (lines 5-9). By line $6$ and induction assumption we have $$\C(R_i) = \{(r,[i-1]), (r_{i+1}, [i-1]), (r_{i},[i-1]), (r_{i-1}, [i-1])\}.$$ Assume first that $q_t = r$.
	Recall that for any $u \in [1,m^2]$ we have $d(r,r_{u}) \geq 2^{u+1} $. It follows that
	\begin{align*}
	R_{i-1} &= \{ r' \in \C(R_i) \mid d(r, r') \leq  2^{i} + 2^{i+1}\} \\
	&= \{(r,[i-1]),(r_{i},[i-1]),(r_{i-1},[i-1])\}
	\end{align*}
	Let us now consider case $q_t \neq r$. We have $d(r,q_t) = 2^{t+1}$ and $d(q_t, r_{u+1}) = 2^{t+1} + 2^{u+2}$ for any $u \in [1,m^2+1]$. Therefore
	$$R_{i-1} = \{ r' \in C_{i-1} \mid d(q_t, r') \leq d(q_t,r) + 2^i + 2^{i+1} \leq 2^{t+1} + 2^{i} + 2^{i+1}\}.$$
	It follows that $R_{i-1} = \{(r,[i-1]), (r_{i},[i-1]), (r_{i-1},[i-1])\}$. In both cases we proceed to line $8$ where we launch FindAllNN$(i-1 , i-1, q_{t}, R_{i-1})$. After proceeding into the recursion we have $j = i$ and therefore query-expansion (lines 9-11) will be performed.  Note that $q_{t}$ was chosen so that $t \geq i+1$. Since every $q_{t-1}$ is either a child of $r$ or $q_{t}$ it follows that FindAllNN$(i-1 , i-2, q_{t'}, R_{i-1})$ will be called for all $t' \geq t-1 \geq i$. Then condition (2) of the induction claim holds as well. 

	\medskip
	\noindent
	It remains to show that Algorithm \ref{alg:cover_tree_k-nearest_dt_original} $(q, R_i = \{r\})$ has $O(m^4)$ low bound on the number of times reference expansions (lines 5-9) are performed.
	Let $\xi$ be the number of times Algorithm \ref{alg:cover_tree_k-nearest_dt_original} performs 
	reference expansions.  For every $q' \in Q$ denote $\xi(q')$ to be the total number of reference expansions performed for $q'$. Recall that any query node $q'$ is introduced in the query expansion (lines 9-11) for parameters $(i = u+1, j = u+1, p, R_i)$, where $p$ is the parent node of $q'$. Since $R_i$ is non empty for all levels $[1,u]$ we have $\xi(q_u) \geq u - 1$ for all $u$. Thus 
	$$\xi = \sum_{q' \in Q} \xi(q') \geq \sum^{m^2+1}_{u = 2} u-2 = O(m^4).$$
	There are different interpretations for the number of duplications. Note that the query tree $\T(Q)$ has exactly one new child on every level and that trees $\T(Q)$ and $\T(R)$ contain exactly the same levels. By using the definitions the number of duplications should be $2$. However, since there can be other interpretations for the number of duplications, we make a rough estimate that the number of duplications is upper bounded by the number of nodes in query tree $O(m^2)$. By using Inequality (\ref{eqa:AuthorClaimingDualTreeKNN}), we obtain the following contradiction: \begin{align*}
	  O(m^4) &= \xi \leq  \max_{p \in R}D(p) \cdot (\max_i \C(R_i)) \cdot (\text{number of duplications}) \\
	  &\leq (2m+1) \cdot 4 \cdot m^2 \leq O(m^3)
	\end{align*}
\end{cexa}

\section{Conclusions and further work}

The motivations for this paper were the past gaps in the proofs of time complexities in \cite[Theorem~5]{beygelzimer2006cover}, \cite[Theorem~6]{beygelzimer2006cover}, \cite[Theorem~3.1]{ram2009linear}, \cite[Theorem~5.1]{march2010fast}. 
In this paper, Example \ref{exa:tall_imbalanced_tree} introduced a dataset $R$ and its cover tree $\T(R)$, where each node appears in a separate level, so the tree 
is split into $\sqrt{|R|}$ different branches and its maximal depth $D(\T(R))$ is $2\sqrt{|R|}$.
\medskip

\noindent 
Counterexample \ref{cexa:construction_algorithm_of_original_cover_tree} shows that the proof of the time complexity \cite[Theorem~6]{beygelzimer2006cover} for the Insert() operation \cite[Algorithm~2]{beygelzimer2006cover} is incorrect for the explicit cover tree $\T(R)$ in Example \ref{exa:tall_imbalanced_tree}. 
Similarly, \cite[Theorem~5]{beygelzimer2006cover} giving time complexity bound for
the nearest neighbors search algorithm \cite[Algorithm~1]{beygelzimer2006cover} has a similar gap in the proof when used on $\T(R)$.
Counterexample~\ref{cexa:dualtreeproof} shows that the same mistake was later repeated in the dual-tree approach for all nearest neighbor search \cite[Theorem~3.1]{ram2009linear}.
\medskip

\noindent
Another forthcoming paper based on the PhD thesis \cite{2205.10194} studies a new compressed cover tree that overcomes the past obstacles in \cite[Theorem~5]{beygelzimer2006cover},  \cite[Theorem~6]{beygelzimer2006cover} and proves the parameterized near-linear time complexities for the compressed cover tree construction and the $k$-nearest neighbor search for any $k\geq 1$. In \cite[Corollary~3.5.11]{2205.10194} it is shown that a compressed cover tree can be constructed in $O(c(R)^{10} \cdot \log_2(|R|) \cdot |R|)$ and \cite[Theorem~3.7.14]{2205.10194} shows that using compressed cover tree $k$-nearest neighbors of any point $q$ can be found in time  $$O(c(R \cup \{q\})^3 \cdot \log_2(k) \cdot (c(R \cup \{q\})^9 \cdot \log_2(|R|) + k)).$$
\medskip

\noindent
The near-linear complexities above have helped justify the fast neighbor-based algorithms for computing isometry invariants of periodic crystals \cite{widdowson2022average,widdowson2021pointwise,anosova2021introduction,anosova2021isometry,anosova2022algorithms}
\medskip

\noindent
We thank all reviewers for their valuable time and suggestions.

\bibliographystyle{abbrv-doi}

\bibliography{template}
\end{document}